\documentclass{LMCS}

\def\dOi{10(3:5)2014}
\lmcsheading%
{\dOi}
{1--18}
{}
{}
{Nov.~\phantom06, 2012}
{Aug.~19, 2014}
{}

\ACMCCS{[{\bf Theory of computation}]: Computational complexity and
  cryptography---Complexity classes}

\usepackage{epsf}
\usepackage{latexsym,color}
\usepackage{hyperref}
\usepackage{amsbsy}
\usepackage{amssymb}
\usepackage{amsmath,url}

\usepackage{latexsym}
\usepackage{graphicx}
\usepackage{subfigure}

\newcommand{\cal}{\mathcal}
\def\DTIME{{\sf DTIME}}
\newcommand{\N}{\mathbb{N}}

\renewcommand{\S}{\mathcal{S}}

\newcommand{\ts}{{t^{*}}}
\newcommand{\G}{\mbox{$\cal G$}}

\newcommand{\Poly}{\mbox{\rm poly}}

\newcommand{\typesetComplexityClass}[1]{\mathsf{#1}}
\newcommand{\DTTR}{\typesetComplexityClass{DTTR}}
\newcommand{\DTTRtime}{\typesetComplexityClass{TTRT}}
\newcommand{\DTTRtimeb}{\typesetComplexityClass{TTRT'}}

\newcommand{\K}{\mbox{\rm K}}

\newcommand{\PA}{\mbox{\rm PA}}

\newcommand{\NEXP}{\mbox{\rm NEXP}}

\newcommand{\EXP}{\mbox{\rm EXP}}

\newcommand{\PSPACE}{\mbox{\rm PSPACE}}
\newcommand{\BPP}{\mbox{\rm BPP}}

\newcommand{\Ppoly}{\mbox{\rm P/poly}}

\renewcommand{\P}{{\rm P}}

\newcommand{\RK}{{R_K}}
\newcommand{\RKtime}{{R_{K^t}}}
\newcommand{\Ktime}{{K^t}}

\newcommand{\RKU}{{R_{K_U}}}

\newcommand{\NP}{\mbox{\rm NP}}

\makeatletter
\newcommand{\comment}[1]{{}}

\newcommand{\Polytime}{\mbox{\rm P}}

\newtheorem{theorem}{Theorem}

\newtheorem{corollary}[theorem]{Corollary}

\makeatother

\title{Reductions to the set of random strings:\\the resource-bounded case}

\author[E.~Allender]{Eric Allender\rsuper a}
\address{{\lsuper a}Department of Computer Science, Rutgers University, Piscataway, NJ 08855, USA}
\email{allender@cs.rutgers.edu}

\author[H.~Buhrman]{Harry Buhrman\rsuper b}%
\address{{\lsuper b}Centrum Wiskunde \& Informatica (CWI), and University of Amsterdam,\hfill\break Amsterdam, The Netherlands}
\email{buhrman@cwi.nl}

\author[L.~Friedman]{Luke Friedman\rsuper c}
\address{{\lsuper c}Google, 1600 Amphitheatre Parkway, Mountain View, CA 94043}
\email{lbfried@gmail.com}

\author[B.~Loff]{Bruno Loff\rsuper d}
\address{{\lsuper d}Centrum Wiskunde \& Informatica (CWI), Amsterdam, The Netherlands}
\email{bruno.loff@gmail.com}

\begin{document}

\begin{abstract}
  This paper is motivated by a conjecture \cite{cie,adfht} that $\BPP$ can be characterized in terms of polynomial-time nonadaptive reductions to the set of Kolmogorov-random strings.  In this paper we show that an approach laid out in \cite{adfht} to settle this conjecture cannot succeed without significant alteration, but that it does bear fruit if we consider time-bounded Kolmogorov complexity instead.

  We show that if a set $A$ is reducible in polynomial time to the set of time-$t$-bounded Kolmogorov-random strings (for all large enough time bounds $t$), then $A$ is in $\Ppoly$, and that if in addition such a reduction exists for any universal Turing machine one uses in the definition of Kolmogorov complexity, then $A$ is in $\PSPACE$.
\end{abstract}

\maketitle

\section{Introduction}

The roots of this investigation stretch back to the discovery that $\PSPACE \subseteq \Polytime^R$ and $\NEXP \subseteq \NP^R$, where $R$ is the set of Kolmogorov-random strings \cite{power,abk}.  Later, it was shown that $\BPP \subseteq \Polytime^R_{tt}$ \cite{bfkl}, where $\Polytime^A_{tt}$ denotes the class of problems reducible to $A$ via polynomial-time {\em nonadaptive} (or {\em truth-table}) reductions.

There is evidence indicating that some of these inclusions are in some sense optimal.  
The reader may, with some justification, be rather confused by this
claim of ``optimality.''  After all, the inclusions in question
all take the form of providing upper bounds for complexity classes,
in terms of efficient reductions to sets such as $R$ that are not even 
computable!  In what sense can these inclusions be optimal?  Let us
explain.

The inclusions mentioned in the initial paragraph hold for the two most-common versions of Kolmogorov complexity (the plain complexity $C$ and the prefix-free complexity $K$), and (significantly for our investigation) they also hold 
{\em no matter which universal Turing machine one uses} when defining the measures $K$ and $C$.

Let $\RKU$ denote the set of random strings according to the prefix-free measure $K$ given by the universal machine $U$: $\RKU = \{x : K_U(x) \geq |x|\}$.  
In a preceding paper \cite{eccc}, it was shown that the class of decidable 
sets that are polynomial-time truth-table reducible to $\RKU$ for every $U$ is contained in $\PSPACE$.  That is, although $\Polytime^{\RKU}_{tt}$ contains arbitrarily complex decidable sets, an extremely complex set can only be there because of characteristics of $\RKU$ that are fragile with respect to the choice of $U$.

This motivates the following definition: $\DTTR$ is the class of all problems 
that are polynomial-time truth-table reducible to ${\RKU}$ for every choice 
of universal prefix-free Turing machine $U$.\footnote{In the conference version
of this paper \cite{mfcspaper}, $\DTTR$ was defined as the class of all
{\em decidable} problems that are polynomial-time truth-table reducible to 
${\RKU}$ for every universal prefix machine $U$.  However, it has recently
been shown that this class remains the same, even if the restriction to 
decidable sets is removed \cite{cai.et.al}.  That is: all sets in $\DTTR$
(as defined above) are already decidable.}
Thus it was proven that
\begin{equation}
  \label{eq:dttr-inclusions}
  \BPP \subseteq \DTTR \subseteq \PSPACE \subseteq \P^\RK.
\end{equation}
So we  naturally come upon the following.

\indent\textbf{Research question:} \emph{Does $\DTTR$ sit closer to $\BPP$, or closer to $\PSPACE$?}

A conjecture by various authors \cite{adfht,cie} is that $\DTTR$ actually characterizes $\BPP$ exactly. Part of the intuition is that (seemingly) a non-adaptive reduction cannot make use of queries to $\RK$ larger than $O(\log n)$ to solve a decidable problem. If this conjecture is indeed true, then we could use the strings of length at most $O(\log n)$ as advice and answer the larger queries with NO, to show that these sets are in $\Ppoly$. The rest of the intuition is that the smaller strings can only be used as a source for pseudo-randomness.
If we are able to prove this conjecture, then we can make use of the tools of Kolmogorov complexity to study various questions about the class $\BPP$. Because of the inclusions listed in (\ref{eq:dttr-inclusions}) above, this now amounts to understanding the relative power of Turing reductions \emph{vs.} truth-table reductions to $\RK$.

In an attempt to tackle this question, it was conjectured in \cite{adfht,cie} that the $\DTTR \subseteq \PSPACE$ upper bound can be improved to $\PSPACE \cap \Ppoly$, and an approach was suggested, based on the above mentioned intuition, dealing with the provability of true statements in various formal systems of arithmetic.  In this paper, we show that this approach must fail, or at least requires significant changes.  Interestingly, we can also prove that this intuition --- that the large queries can be answered with NO --- \emph{can} be used  in the resource-bounded setting to show an analogue of the $\Ppoly$ inclusion.  While demonstrating this discrepancy we show several other ways in which reductions to $R_K$ and $R_{K^t}$ are actually very different; in particular, we construct a counter-intuitive example of a polynomial-time non-adaptive reduction that distinguishes $R_K$ from $R_{K^t}$, for any sufficiently large time-bound $t$.

To investigate the resource-bounded setting we define a class $\DTTRtime$ as 
an analog of $\DTTR$, defined using time-bounded Kolmogorov complexity (for
very large time bounds).
Informally, $\DTTRtime$ is the class of problems that are polynomial-time truth-table reducible to ${\RKtime}$ for every sufficiently fast-growing time-bound $t$, and every ``time-efficient'' universal Turing machine used to define $\Ktime$.
We prove that, for all monotone nondecreasing computable functions $\alpha(n) = \omega(1)$,
$$\BPP \subseteq \DTTRtime \subseteq \PSPACE/\alpha(n) \cap \Ppoly.$$
Here, $\PSPACE/\alpha(n)$ is a ``slightly non-uniform'' version of $\PSPACE$.
That is, we succeed in obtaining a $\Ppoly$ upper bound (of the sort that
we were unable to obtain for $\DTTR$ in \cite{adfht}), and we ``nearly''
obtain a $\PSPACE$ upper bound (analogous to the $\PSPACE$ upper bound that
was obtained for $\DTTR$ in \cite{eccc}).
We believe that this indicates that $\DTTRtime$ is ``closer'' to $\BPP$ than it is to $\PSPACE$.
(Recently, Hirahara and Kawamura have announced results of a similar
nature, stated in terms of plain Kolmogorov complexity, instead of
the prefix-free notion considered here \cite{mfcs14}.)

It would be more appealing to avoid the advice function, and we are able to do so, although this depends on a fine point in the definition of time-efficient prefix-free Kolmogorov complexity.  This point involves a subtle technical distinction, and will be left for the appropriate section. To summarize:
\begin{itemize}
\item In Section \ref{sec:small-circuits} we prove that $\DTTRtime \subseteq \Ppoly$, by using the same basic idea of \cite{adfht,cie}. We further show, however, that this approach will not work to prove $\DTTR \subseteq \Ppoly$, and by reversing the logic connection of \cite{adfht,cie}, this will give us an independence result in certain extensions of Peano arithmetic.
\item Then in Section \ref{sec:pspace-inclusion} we prove that $\DTTRtime \subseteq \PSPACE/\alpha(n)$, which is a non-trivial adaptation of the techniques from \cite{eccc}. In Section \ref{sec:no-advice} we show how to get an analogous result without the super-constant advice term.
\end{itemize}
In the final section we discuss prospects for future work.

We consider the results in Section \ref{sec:small-circuits} to be the most 
important contributions of this paper.  The $\Ppoly$ upper bound
indicates that $\DTTRtime$ is a ``feasible'' class in some sense, and
can perhaps be viewed as evidence that a similar upper bound should
also hold for $\DTTR$ -- while simultaneously showing that rather different
techniques will be required to establish such a bound for $\DTTR$.  If
such a bound can be proved, then this would show $\BPP \subseteq \DTTR
\subseteq \PSPACE \cap \Ppoly$, which would in turn be a significant step
toward proving $\DTTR = \BPP$.  We refer the reader to \cite[Section 7]{adfht},
for additional motivation for studying the $\DTTR = \BPP$ question.

In contrast, the results in Sections \ref{sec:pspace-inclusion} and 
\ref{sec:no-advice} are more technical, and draw heavily on the
techniques of \cite{eccc}, in order to prove an upper bound for
$\DTTRtime$ that is similar to the bound proved for $\DTTR$ in
\cite{eccc}.

\section{Preliminaries}\label{prelimsec}

We assume the reader is familiar with basic complexity theory \cite{balcazar_dg:88} and Kolmogorov complexity \cite{LV}. We use $\le^p_T$ and $\Polytime^A$ when referring to polynomial-time Turing reductions, and $\le^p_{tt}$ and $\Polytime^A_{tt}$ for polynomial-time truth-table (or \emph{non-adaptive}) reductions.
For example, $M : A \le^p_{T} B$ means that $M$ is a polynomial-time Turing reduction from $A$ to $B$.
For a set $A$ of strings, $A^{\leq n}$ denotes the set of all strings of length at most $n$ in $A$.

We let $K_U$ denote Kolmogorov complexity with respect to prefix machine $U$, 
i.e., $K_U(x) = \min \{|p| : U(p) = x\}$.  (A prefix machine is a Turing 
machine with the property that, if it halts and produces output on some
input $p$, then it does {\em not} halt on any input of the form
$px$, for any nonempty string $x$.  For more details, see \cite{LV}.)
We use $R_{K_U}$ to denote the set of \emph{$K_U$-random} strings $\{ x | K_U(x) \ge |x| \}$. In this paper, a function $t:\N\to\N$ is called a ``time-bound'' if it is non-decreasing and time-constructible. 
(We follow the usual convention that a time-constructible function $t$ 
satisfies $t(n) \geq n$ for all $n$.  See \cite{balcazar_dg:88}.)
We use the following time-bounded version of Kolmogorov complexity: for a prefix machine $U$ and a time-bound $t$, $K_U^t(x)$ is the length of the smallest string $p$ such that $U(p)$ outputs $x$ and halts in fewer than $t(|x|)$ time steps. Then $R_{K^t_U}$ is the set of \emph{$K^t_U$-random} strings $\{ x | K_U^t(x) \ge |x| \}$.
Let us define what it means for a machine to be ``universal'' in the time-bounded setting:

\begin{defi}\label{timeeffdef}
  A prefix machine $U$ is a time-efficient universal prefix machine if there exist constants $c$ and $c_M$ for each prefix machine $M$, such that
  \begin{enumerate}
  \item $\forall x$, $K_U(x) \leq K_M(x) + c_M$, and
  \item $\forall x$, and for all time bounds $t$ and $t'$ where
$t > {t^\prime}^c$, $K^t_U(x) \leq K^{t'}_M(x) + c_M$.
  \end{enumerate}
\end{defi}

We will sometimes omit $U$ in the notation $K_U, R_{K_U}, K_U^t$, $R_{K_U^t}$, in which case we mean $U = U_0$, for some arbitrary choice of a time-efficient universal prefix machine $U_0$. Now we can formally define the time-bounded analogue of $\DTTR$:

\begin{defi}
$\DTTRtime$ is the class of languages $L$ such that there exists a time bound $t_0$ (depending on $L$) such that for all time-efficient universal prefix machines $U$ and for all time-constructible $t \geq t_0$, $L \leq_{tt}^p R_{K^t_U}$.
\end{defi}

Clearly, every language in $\DTTRtime$ is decidable.


The proof of Corollary 12 in \cite{bfkl} shows that, if $t \ge t_0 = 2^{2^{2 n}}$, then $\BPP \le^p_{tt} R_{K_U^t}$, for any time-efficient universal $U$. This implies:

\begin{thm}[\cite{bfkl}]
  $\BPP \subseteq \DTTRtime$.
\end{thm}

Now we prove some basic facts about time-bounded prefix-free Kolmogorov complexity.

\begin{prop}\label{simpprop}
  For any machine $M$ and $t'(|x|) > 2^{|x|} t(|x|)$, the query $x \in R_{K_M^t}?$ can be answered in time $t'$.
\end{prop}

\begin{proof}
  Simulate the machine $M$ on every string of length less than $|x|$ for $t(|x|)$ steps.  Because there are fewer than $2^{|x|}$ such strings, the bound follows.
\end{proof}

\begin{prop}\label{simpprop2}
Let $L \leq_{tt}^p R_{K^t_U}$ for some time-bound $t$.
Then there exists a constant $k$ such that the language $L$ can be decided in  $t_L(n) = 2^{n^k}t(n^k)$ time.
\end{prop}

\begin{proof}
Let $M$ be a machine that decides $L$ by running the polynomial-time truth-table reduction from $L$ to $R_{K^t_{U}}$, and computing by brute-force the answer to any queries from the reduction.  Using Proposition \ref{simpprop}, we have that for large enough $k$, $M$ runs in at most $t_L(n) = 2^{n^k} t(n^k)$ time, so $L$ is decidable within this time-bound.
\end{proof}

It is the ability to compute $R_{K^t}$ for short strings that makes the time-bounded case different from the ordinary case. This will be seen in proofs throughout the paper.

\section{How and why to distinguish $R_K$ from $R_{K^t}$}\label{sec:small-circuits}


At first glance, it seems reasonable to guess that a polynomial-time reduction would have difficulty telling the difference between an oracle for $\RK$ and an oracle for $\RKtime$, for large enough $t$. Indeed $\RK \subseteq \RKtime$ and in the limit for $t \rightarrow \infty$ they coincide.

One might even suspect that a polynomial-time reduction must behave the same way with $R_{K^t}$ and $R_K$ as oracle, already for  modest time bounds $t$.
However, this intuition is wrong.  Here is an example for adaptive polynomial-time reductions.

\begin{obs}
  There is a polynomial-time algorithm which, given oracle access to $\RK$ and input $1^n$, outputs a $K$-random string of length $n$. However, for any time-bound $t$ such that
$t(n+1) \gg 2^n t(n),$
  there is no polynomial-time algorithm which, given oracle access to $\RKtime$ and input $1^n$, outputs a $\K^t$-random string of length $n$.
\end{obs}

For the algorithm, see \cite{buhrman:2005}; roughly, we start with a small random string and then use 
\cite[Theorem 15]{buhrman:2005} (described later)
to get a successively larger random string. But in the time-bounded case in~\cite{BuhrmanM97} it is shown that on input $1^n$, no polynomial-time machine $M$ can query (or output) any $K^t$-random string of length $n$: in fact, $M(1^n)$ is the same for both oracles $R_{K^t}$ and $R' = R_{K^t}^{\le n - 1}$. This is proven as follows: since $R'$ can be computed in time $t(n)$ (by Proposition \ref{simpprop}), then any query of length $\ge n$ made by $M^{R'}(1^n)$ is described by a pointer of length $O(\log n)$ in time $t(n)$, and hence is not in $R_{K^t}$.

\subsection{Small circuits for sets reducible to $R_{K^t}$}

We now prove that $\DTTRtime$ is a subset of $\Ppoly$. Actually, we will prove that this holds even for {\em Turing} reductions to $\RKU$, (for {\em every} universal Turing machine $U$):

\begin{thm}\label{ppolythm}
  Suppose $A \in \DTIME(t_1)$ and $M:A \le^p_{T} R_{K^t}$, for some time-bounds $t, t_1$ with $t(n+1) \ge 2^n t(n) + 2^{2^n} t_1(2^n)$.\footnote{For example, if $A \in \EXP$, then $t$ can be doubly-exponential. If $A$ is elementary-time computable, then $t$ can be an exponential tower.} Then $A \in \Ppoly$; in fact, if $M$ runs in time $n^c$, and $R' = R_{K^t}^{\le \lceil (c +1) \log n\rceil}$, then $\forall x \in \{ 0, 1 \}^n \; M^{R'}(x) = A(x)$.
\end{thm}

\begin{proof}
Let $\ell(n) = \lceil (c + 1) \log n \rceil$, and let 
$R'(n) = R_{K^t}^{\le \ell(n)}$.  Showing that $M^{R'}(x) = A(x)$ for all
$x$ of each length $n$ suffices to show that $A \in \Ppoly$, because
$R'$ consists of only polynomially-many strings, which can be encoded in
an advice string of length polynomial in $n$.

Suppose, for the sake of contradiction, that
$M^{R'(n)}(x) \not= A(x)$ for some $x$ of length $n$. Then we may find the first such $x$ in time $2^{\ell(n)} t(\ell(n)) + 2^{n+1} (t_1(n) + O(n^c))$ (cf. Proposition \ref{simpprop}), and each query made by $M^{R'(n)}(x)$ can be output by a program of length $c \log n + O(1)$, running in the same time bound. But since $A(x) \not= M^{R'(n)}(x)$, it must be that, with $R'(n)$ as oracle, $M$ makes some query $q$ of size $m \ge \ell(n) + 1$ which is random for $t$-bounded Kolmogorov complexity (because both small and nonrandom queries are answered correctly when using $R'$ instead of $R_{K^t}$). Hence we have both that $q$ is supposed to be random, and that $q$ can be output by a program of length $< \ell(n)$ in time $2^{\ell(n)} t(\ell(n)) + 2^{n+1} (t_1(n) + O(n^c)) \ll 2^{\ell(n)} t(\ell(n)) + 2^{2^{\ell(n)}} t_1(2^{\ell(n)}) \le t(\ell(n) + 1) \le t(m)$, which is a contradiction.
\end{proof}

\begin{corollary}\label{ttrtcor}
$\DTTRtime \subseteq \Ppoly$.
\end{corollary}

\begin{proof}
Let $L \in \DTTRtime$.
By the definition of $\DTTRtime$, $L \leq_{tt}^p R_{K^{t_0}}$ for some $t_0$.
Using Proposition \ref{simpprop2}, we then have that $L$ is decidable in time $t_L(n) = 2^{n^k}t_0(n^k)$ for some constant $k$.
Choose a time-bound $t$ such that $t(n+1) \ge 2^n t(n) + 2^{2^n} t_L(2^n)$.
By the definition of $\DTTRtime$, since $t > t_0$, we have that $L \leq_{tt}^p R_{K_{U_0}^{t}}$, from which by Theorem \ref{ppolythm} it follows that $L \in \Ppoly$.
\end{proof}

$\PSPACE \le^p_T R_K$ \cite{power}, but Theorem \ref{ppolythm} implies that $\PSPACE \not\le^p_T R_{K^t}$ for sufficiently-large $t$, unless $\PSPACE \subseteq \Ppoly$. This highlights the difference between the time-bounded and ordinary Kolmogorov complexity, and how this comes to the surface when working with reductions to the corresponding sets of random strings.
We wish to emphasize at this point that the proof of the inclusion
$\PSPACE \le^p_T R_K$ relies on the ability of a $\Polytime^{R_K}$ computation
to construct a large element of $R_K$, whereas the $\Ppoly$ upper bound in
the time-bounded case relies on the {\em inability} to use the oracle to
find such a string, in the time-bounded setting.
 

\subsection{A reduction distinguishing $R_K$ from $R_{K^t}$, and an incorrect conjecture}\label{pasec}

Theorem \ref{ppolythm} shows that a polynomial-time truth-table reduction to $R_{K^t}$ for sufficiently-large $t$ will work just as well if only the logarithmically-short queries are answered correctly, and all of the other queries are simply answered ``no''.

The authors of \cite{adfht} conjectured that a similar situation would hold if the oracle were $\RK$ instead of $R_{K^t}$.  More precisely, they proposed a 
proof-theoretic approach towards proving that $\DTTR$ is in $\Ppoly$: Let $\PA_0$ denote Peano Arithmetic, and for $k>0$ let $\PA_k$ denote $\PA_{k-1}$ augmented with the axiom ``$\PA_{k-1}$ is consistent''.  In \cite{adfht} it is shown that, for any polynomial-time truth-table reduction $M$ reducing a decidable set $A$ to $\RK$, one can construct a true statement of the form $\forall n \forall j \forall k \Psi(n,j,k)$ (which is provable in a theory such as Zermelo-Frankel), with the property that if, for each fixed ({\bf n},{\bf j},{\bf k}) there is some $k'$ such that $\PA_{k'}$ proves $\psi$({\bf n},{\bf j}, {\bf k}), then $\DTTR \subseteq \Ppoly$.  Furthermore, if these statements were provable in the given extensions of $\PA$, it would follow that, for each input length $n$, there is a finite subset $R' \subseteq \RK$ consisting of strings having length at most $O(\log n)$, such that $M^{R'}(x) = A(x)$ for all strings $x$ of length $n$.

Thus the authors of \cite{adfht} implicitly conjectured that, for any polynomial-time truth-table reduction of a decidable set to $\RK$, and for any $n$, there would be some setting of the short queries so that the reduction would still work on inputs of length $n$, when all of the long queries are answered ``no''.  While we have just seen that this is precisely the case for the time-bounded situation, the next theorem shows that this does not hold for $\RK$, even if ``short'' is interpreted as meaning ``of length $< n$''.
(It follows that infinitely many of the statements $\psi$({\bf n},{\bf j}, {\bf k}) of \cite{adfht} are independent of every $\PA_{k'}$.)

\begin{thm}\label{thm:distinguish}
  There is a truth-table reduction $M: \{ 0, 1 \}^* \le^p_{tt} R_K$, such that, for all large enough $n$:
  \[
  \forall R' \subseteq \{ 0, 1 \}^{\le n-1} \exists x \in \{ 0, 1 \}^n \; M^{R'}(x) \not= 1.
  \]
\end{thm}

%

\begin{proof}  
Theorem 15 of \cite{buhrman:2005} presents a polynomial-time procedure which, 
given a string $z$ of even length $n-2$, will output a list of constantly-many 
strings $z_1, \ldots, z_c$ of length $n$, such that at least one of them will 
be $K$-random if $z$ is.  
We use this to define our reduction $M$ as follows: on input $x = 00\ldots0 z$ of length $n$ having even $|z|$, we query each of $z$, $z_1, \ldots, z_c$, and every string of length at most $\log n$. If there are no 
  strings of length at most $\log n$ in the oracle, we reject. Else, if $z$ is in the oracle but none of the $z_i$ are, we reject. On all other cases we accept.

By \cite[Theorem 15]{buhrman:2005}, and since $R_K$ has strings at every length, it is clear that $M$ accepts every string with oracle $R_K$, and rejects every string if $R' = \varnothing$.  However, for any non-empty set $R' \subseteq \{ 0, 1 \}^{\le n-1}$, let $\ell\le n-1$ be the highest even length for which ${R'}^{=\ell} \not= \varnothing$, and pick $z\in {R'}^{=\ell}$.  Then we will have $z \in {R'}^{=\ell}$ but every $z_i \not\in R^{=\ell + 2}$, hence $M^{R'}(00\ldots0 z)$ rejects.
\end{proof}

In fact, if we let $R' = R_{K^t}^{\le n - 1}$, for even $n$, then for the first $x = 00 z$ such that $M^{R'}(x) = 0$, we will have $z \in R' \subseteq R_{K^t}$, but each $z_i$ can be given by a small pointer in time $O(2^{n-1} t(n-1))$ (again we use Proposition \ref{simpprop}), and hence $z_i \not\in R_{K^t}$ for suitably fast-growing $t$. Thus $M^{R_{K^t}}(x) = 0 \not= M^{R_K}(x)$, and we conclude:

\begin{obs}
  If $t(n+1) \gg 2^n t(n)$, then the non-adaptive reduction $M$ above behaves differently on the oracles $R_K$ and $R_{K^t}$.
\end{obs}

\section{Polynomial Space with Advice}\label{sec:pspace-inclusion}

Our single goal for this section is proving the following:

\begin{thm}\label{main}
  For any computable unbounded function $\alpha(n) = \omega(1)$, $$\DTTRtime \subseteq \PSPACE / \alpha(n).$$
\end{thm}

The proof of this theorem is patterned closely on related arguments in \cite{eccc}, although a number of complications arise in the time-bounded case.  
Although we aim to make the presentation here self-contained, \cite{eccc} is a good primer and a source of additional intuition for the proof.
Also, one can refer to the conference version of this paper \cite{mfcspaper} for a presentation that is not self-contained but emphasizes the differences between the proof in the time-bounded case and the unbounded case.
Before proving the theorem we present several supporting propositions.

\begin{prop}\label{prefix_entr}
For any time bound $t$ and time-efficient universal prefix machine $U$, 
$$\sum_{x \in \{0,1\}^*} 2^{-K^t_U(x)} \leq 1.$$
\end{prop}

\begin{proof}
From the Kraft Inequality (see e.g. \cite{LV}, Theorem 1.11.1),  $\sum_{x \in \{0,1\}^*} 2^{-K_U(x)} \leq 1$ for any prefix machine $U$.
For any time bound $t$ and string $x$, $K_U^t(x) \geq K_U(x)$, so adding a time bound can only decrease the sum on the left side of this inequality.
\end{proof}

\begin{prop}[Analogue to Coding Theorem]\label{coding}
  Let $f$ be a function such that
  \begin{enumerate}
  \item $\sum_{x \in \{0,1\}^{*}} 2^{-f(x)} \leq 1$, and
  \item there is a machine $M$ computing $f(x)$ in time $t(|x|)$.
  \end{enumerate}
Let $t^\prime(|x|) > {2^{2|x|}} t(|x|)$.  Then for some $M^\prime$, $K^{t^\prime}_{M^\prime}(x) = f(x)+2$.
\end{prop}

\begin{proof}

  The proof is similar to the proof of Proposition 5 from \cite{eccc}.  Let $$E = \langle x_0, f(x_0) \rangle, \langle x_1, f(x_1) \rangle, \ldots$$ be an enumeration of the function $f$ ordered lexicographically by the strings $x_i$.

  We identify the set of infinite sequences $S = \{0,1\}^\infty$ with the half-open real interval $[0,1)$; that is, each real number $r$ between 0 and 1 will be associated with the sequence(s) corresponding to the infinite binary expansion of $r$.  We will associate each element $\langle x_i,f(x_i) \rangle$ from the enumeration $E$ with a subinterval $I_i \subseteq S$ as follows:

  $I_0 = [0,2^{-f(x_0)})$, and for $i \geq 1$, $I_i = [\sum_{k < i}2^{-f(x_k)}, \sum_{k\leq i}2^{-f(x_k)})$.  That is, $I_i$ is the half-open interval of length $2^{-f(x_i)}$ that occurs immediately after the interval corresponding to the element $\langle x_{i-1},f(x_{i-1}) \rangle$ that appeared just prior to $\langle x_i, f(x_i) \rangle$ in the enumeration $E$.

  Since $\sum_{i \geq 0} 2^{-f(x_i)} \leq 1$, each $I_i \subseteq S$.

  Any {\em finite} string $z$ also corresponds to a subinterval $\Gamma_z \subseteq S$ consisting of all infinite sequences that begin with $z$; $\Gamma_z$ has length $2^{-|z|}$.  Given any element $\langle x_i, f(x_i) \rangle$, there must exist a lexicographically first string $z_i$ of length $f(x_i)+2$ such that $\Gamma_{z_i} \subseteq I_i$.  Observe that, since the intervals $I_i$ are disjoint, no string $z_i$ is a prefix of any other.

  Let $M^\prime$ be the following machine.  On input $z$, $M^\prime$ runs $M$ to compute the enumeration $E$ until it finds an element $\langle x_i, f(x_i) \rangle$ that certifies that $z = z_i$.  If it finds such an element then $M^\prime$ outputs $x_i$.

  Suppose that $M^\prime$ outputs $x_i$ on input $z$, and let $\langle x_i, f(x_i) \rangle$ be the element of $E$ corresponding to $x_i$.  Before outputting $x_i$, $M^\prime$ must compute $|\langle x_j, f(x_j) \rangle|$ for every string $x_j$ such that $x_j < x_i$ (under the lexicographical ordering).  There are at most $2^{|x_i| + 1}$ strings $x_j$ such that $x_j < x_i$, so overall this will take less than $2^{2|x_i|} t(|x_i|)$ time.

  $M^\prime$ will be a prefix machine, and we have that $K^{t^\prime}_{M^\prime}(x) = f(x) + 2$.\end{proof}

Given two Kolmogorov complexity functions, their minimum is {\em not} 
necessarily going to be a Kolmogorov complexity function; this is the case
both in the time-bounded setting and in the traditional setting without
time bounds.  But one can come close.  The following proposition establishes
that there is a time-bounded Kolmogorov complexity function that is precisely
one more than the minimum of two other given time-bounded Kolmogorov
complexity functions.

\begin{prop}[Analogue to Proposition 6 from \cite{eccc}]\label{old6}
  Let $U$ be a time-efficient universal prefix Turing machine and let $M$ be any prefix Turing machine. Suppose that $t,t^\prime$, and $t''$ are time bounds and $f,g$ are two time-constructible increasing functions, such that $f$ is upper bounded by a polynomial, and $t^{\prime\prime}(|x|) \geq \max\{f(t(|x|)),g(t^\prime(|x|))\}$.

  Then there is a time-efficient universal prefix machine $U^{\prime}$ such that
  \[K_{U^{\prime}}^{t^{\prime\prime}}(x) = \min(K_U^t(x), K_{M}^{t^\prime}(x)) + 1.\]
\end{prop}

\begin{proof}
  On input $0y$, $U^{\prime}$ runs $U$ on input $y$.  If $U$ would output string $x$ on $y$ after $s$ steps, then $U^{\prime}$ outputs string $x$ after $f(s)$ steps.  Similarly, on input $1y$, $U^\prime$ runs $M$ on input $y$.  If $M$ would output string $x$ on $y$ after $s$ steps, then $U^{\prime}$ outputs string $x$ after $g(s)$ steps.

  Note that because $U$ is an efficient universal prefix machine, $U^\prime$ will be an efficient universal prefix machine as well. 
\end{proof}

\begin{prop}[Analogue of Proposition 7 from \cite{eccc}]\label{old7}
  Given any time-efficient universal prefix machine $U$, time bound $t$, and constant $c \geq 0$, there is a time-efficient universal prefix machine $U^\prime$ such that $K^t_{U^\prime}(x) = K^t_U(x) + c$.
\end{prop}

\begin{proof}
  On input $0^cx$, $M^\prime$ runs $M$ on input $x$, and doesn't halt on other inputs. 
\end{proof}

\begin{proof}[Proof of Theorem \ref{main}]
  Fix $\alpha$, and suppose for contradiction that $L \in \DTTRtime - \PSPACE/\alpha(n)$.  Let $t_0$ be the time bound given in the definition of $\DTTRtime$, and assume without loss of generality that $t_0(n)$ is greater than the time required to compute the length of the advice $\alpha(n)$,  and let $U_0$ be some arbitrary time-efficient universal prefix machine.  By the definition of $\DTTRtime$, $L \leq_{tt}^p R_{K^{t_0}_{U_0}}$.  
Therefore, by Proposition \ref{simpprop2}, $L$ is decidable in time $t_L(n) = 2^{n^k} t_0(n^k)$ for some constant $k$.

Let $\ts(n)$ be an extremely fast-growing time-constructible function, so that for any constant $d$, we have $\ts(\log(\alpha(n))) > 2^{n^d}t_L(n)$ for all large $n$.  To get our contradiction, we will show that there exists a time-efficient universal prefix machine $U$ such that $L \not \leq_{tt}^p R_{K^{\ts^3}_U}$.  Note that because $\ts > t_0$, this is a contradiction to the fact that $L \in \DTTRtime$.

  For any function $f : \{0,1\}^* \rightarrow \N$, define $R_f = \{x : f(x) \geq |x|\}$. We will construct a function $F : \{0,1\}^* \rightarrow \N$ and use it to form a function $H: \{0,1\}^* \rightarrow \N$ such that:
  \begin{enumerate}
  \item $F$ is a total function and $F(x)$ is computable in time $\ts^2(|x|)$ by a machine $M$;
  \item $H(x) = \min(K_{U_0}^\ts(x) + 5, F(x) + 3)$;
  \item $\sum_{x \in \{0,1\}^*} 2^{-H(x)} \leq 1/8$;
  \item $L \not \leq ^p_{tt} R_H$.
  \end{enumerate}
\end{proof}

\begin{clm}[Analogue of Claim 1 from \cite{eccc}]\label{old1}
  Given the above properties $H = K_{U}^{\ts^3}$ for some efficient universal prefix machine $U$.
\end{clm}

By Property 4 this ensures that the theorem holds.

\begin{proof}
  By Property 3 we have that $\sum_{x \in \{0,1\}^*} 2^{-(F(x)+3)} \leq 1/8$.  Hence $\sum_{x \in \{0,1\}^*} 2^{-F(x)} \leq 1$.  Using this along with Property 1, we then have by Proposition \ref{coding} that $K^{\ts^3}_{M^\prime} = F+2$ for some prefix machine $M^\prime$.  By Proposition \ref{old7} we have that $K^\ts_{U'} = K_{U_0}^\ts + 4$ for some efficient universal prefix machine $U'$.  Therefore, by Proposition \ref{old6}, with $f(n) = n^3, g(n) = n$, we find that $H(x) = \min(K^{\ts}_{U_0}(x) + 5, F(x) + 3) = \min(K^{\ts^3}_{M'}, K^{\ts}_{U'}(x)) + 1$ is $K^{\ts^3}_{U}$ for some efficient universal prefix machine $U$. 
\end{proof}
All we now need to show is that, for our given language $L$, we can always construct functions $H$ and $F$ with the four desired properties.  

Let $\gamma_1,\gamma_2,\ldots$ be a list of all possible polynomial-time truth-table reductions from $L$ to $R_H$.
This is formed in the usual way: we take a list of all Turing machines and 
put a clock of $n^i+i$ on the $i$th one and
we will interpret the output on a string $x$ as an encoding of a Boolean circuit
on atoms of the form ``$z \in R_H$''. (i.e. these atoms form the input gates of the circuit, and their truth values determine the output of the circuit.)  We will refer to the string $z$ as a \emph{query}.

As in \cite{eccc}, to ensure that $L \not \leq^p_{tt} R_H$ (Property 4), we need to satisfy an infinite list of requirements of the form\\
\hspace*{.25in}$R_e : \gamma_e \mbox{ is not a polynomial-time truth-table reduction of $L$ to } R_H.$

As part of our construction we will set up and play a number of games, which will enable us to satisfy each of these requirements $R_e$ in turn.
Our moves in the game will define the function $F$ (and thus indirectly $H$).  Originally we have that $F(z) = 2|z|+3$ for all strings $z$. 
Potentially during one of these games, we will play a move forcing a string $z$ to be in the complement of $R_H$.  To do this we will set $F(z) = |z|- 4$.  Therefore, a machine $M$ can compute $F(z)$ by running our construction, looking for the first time during the construction that $F(z)$ is set to $|z| - 4$, and outputting $|z| - 4$.  If a certain amount of time elapses (to be determined later) during the construction without $F(z)$ ever being set to $|z| - 4$, then the machine $M$ outputs the default value $2|z|+3$.

\subsection{Description of the games}

Let us first describe abstractly the games that will be played during the construction; afterwards we will explain how it is that we use these games to satisfy each requirement $R_e$. (Note that these games are defined differently than those in \cite{eccc}).

For a given requirement $R_e$, a game $\G_{e,x}$ will be played as followed for some string $x$:

First we calculate the circuit $\gamma_{e,x}$, which is the output of the reduction $\gamma_e$ on input $x$.  Let $F^*$ be the function $F$ as it is at this point of the construction when the game $\G_{e,x}$ is about to be played.  For any atom ``$z_i \in R_H$'' that is an input of this circuit such that $|z_i| \leq \log(\alpha(|x|))-1$, we calculate $r_i = \min(K^\ts_{U_0}(z_i)+5, F^*(z_i) + 3)$.  If $r_i < |z_i|$ we substitute FALSE in for the atom, and simplify the circuit accordingly, otherwise we substitute TRUE in for the query, and simplify the circuit accordingly.  (We will refer to this as the ``pregame preprocessing phase''.)

The remaining queries $z_i$ are then ordered by increasing length.  There are two players, the $F$ player (whose moves will be played by us during the construction), and the $K$ player (whose moves will be determined by $K_{U_0}^\ts$).  As in \cite{eccc}, in each game the $F$ player will either be playing on the YES side (trying to make the final value of the circuit equal TRUE), or the NO side (trying to make the final value of the circuit equal FALSE).

Let $S_1$ be the set of queries from $\gamma_{e,x}$ of smallest length, let $S_2$ be the set of queries that have the second smallest length, etc.  So we can think of the queries being partitioned into an ordered set $\S = (S_1, S_2, \ldots, S_r)$ for some $r$.

The scoring for the game is similar to that in \cite{eccc}; originally each player has a score of 0 and a player loses if his score exceeds some threshold $\epsilon$.  When playing a game $\G_{e,x}$, we set $\epsilon = 2^{-e-3}$.

Originally we have that the truth value of all the atoms in the game are TRUE.
In round one of the game, the $K$ player makes some (potentially empty) subset $Z_1$ of the queries from $S_1$ nonrandom; i.e. for each $z \in Z_1$ he sets the atom ``$z \in R_H$'' to the value FALSE.  For any $Z_1 \subseteq S_1$ that he chooses to make nonrandom, 
$\sum_{z \in Z_1} (2^{-(|z| - 6)} - 2^{-(2|z|+3)})$ is added to his score.
As in \cite{eccc}, a player can only legally make a move if doing so will not cause his score to exceed $\epsilon$.



After the $K$ player makes his move in round 1, the $F$ player responds, by making some subset $Y_1$ of the queries from $S_1 - Z_1$ nonrandom.  After the $F$ player moves, $\sum_{z \in Y_1} 2^{-(|z| - 4)} - 2^{-(2|z|+3)}$ is added to his score.

This is the end of round one.  Then we continue on to round two, played in the same way.  The $K$ player goes first and makes some subset of the queries from $S_2$ nonrandom (which makes his score go up accordingly), and then the $F$ player responds by making some subset of the remaining queries from $S_2$ nonrandom.  Note that if a query from $S_i$ is not made nonrandom by either the $K$ player or the $F$ player in round $i$, it cannot be made nonrandom by either player for the remainder of the game.

After $r$ rounds are finished the game is done and we see who wins, by evaluating the circuit $\gamma_{e,x}$ using the answers to the queries that have been established by the play of the game.  If the circuit evaluates to TRUE (FALSE) and the $F$ player is playing as the YES (NO) player, then the $F$ player wins, otherwise the $K$ player wins.

Note that the game is asymmetric between the $F$ player and the $K$ player; the $F$ player has an advantage due to the fact that he plays second in each round and can make an identical move for fewer points than the $K$ player.  Because the game is asymmetric, it is possible that $F$ can have a winning strategy playing on {\em both} the YES and NO sides.  Thus we define a set $val(\G_{e,x^\prime}) \subseteq \{0,1\}$ as follows: $0 \in val(\G_{e,x^\prime})$ if the $F$ player has a winning strategy playing on the NO side in $\G_{e,x^\prime}$, and $1 \in val(\G_{e,x^\prime})$ if the $F$ player has a winning strategy playing on the YES side in $\G_{e,x^\prime}$.

\subsection{Description of the construction}

Now we describe the construction.  In contrast to the situation in \cite{eccc}, we do not need to worry about playing different games simultaneously or dealing with requirements in an unpredictable order; we will first satisfy $R_1$, then $R_2$, etc.  To satisfy $R_e$ we will set up a game $\G_{e,x}$ for an appropriate string $x$ of our choice, and then play out the game in its entirety as the $F$ player.  We will choose $x$ so that we can win the game $\G_{e,x}$, and will arrange that by winning the game we ensure that $R_e$ is satisfied.  

A complication that arises is that the player $K$ (whose moves are decided by
$U_0$) is not constrained to make only ``legal'' moves.  That is, player $K$
might decide to make moves that exceed the legal threshold while playing some
of the games.  If the $K$ player ``cheats'' on game $\G_{e,x}$, then we quit 
the game $\G_{e,x}$ and we play $\G_{e,x^\prime}$ for some new $x^\prime$.  
However, we will show that the $K$ player cannot cheat infinitely often on games for a particular $e$, so eventually $R_e$ will be satisfied.

Originally we define the function $F$ so that $F(z) = 2|z|+3$ for all strings $z$.
Suppose $s$ time steps have elapsed during the construction up to this point, and we are getting ready to construct a new game in order to satisfy requirement $R_e$.  (Either because we just finished satisfying requirement $R_{e-1}$, or because $K$ cheated on some game $\G_{e,x}$, so we have to start a new game $\G_{e,x^\prime}$).
Starting with the string $0^{\ts^4(s)}$ (i.e. the string of $\ts^4(s)$ zeros), we search strings in lexicographical order until we find an $x^\prime$ such that $(1 - L(x^\prime)) \in val(\G_{e,x^\prime})$.  (Here, $L$ denotes the characteristic function of the set $L$.)

Once we find such a string $x^\prime$ (which we will prove we always can), then we play out the game $\G_{e,x^\prime}$ with the $F$ player (us) playing on the YES side if $L(x^\prime) = 0$ and the NO side if $L(x^\prime) = 1$.  To determine the $K$ player's move in the $i$th round, we let $Z_i \subseteq S_i$ be those queries $z \in S_i$ for which $K_{U_0}^{\ts}(z) \leq |z| - 6$.  Our moves are determined by our winning strategy;  whenever we play a move that makes a query $z$ nonrandom, we update the function $F$ so that $F(z) = |z| - 4$. 
Note that whenever either of the players plays a move involving a query $z$ in one of the games (which we have called ``making $z$ nonrandom''), he \emph{does} make the query $z$ nonrandom in the sense that $R_H(z)$ is fixed to the value 0 for good.

To finish showing that Property 4 will be satisfied, it suffices to prove the following three claims.

\begin{clm}
If during the construction we win a game $\G_{e,x}$, then $R_e$ will be satisfied and will stay satisfied for the remainder of the construction.
\end{clm}

\begin{proof}
Suppose that we win a game $\G_{e,x}$.  
Let $H^* = \min(K_{U_0}^\ts + 5, F^* + 3)$, where $F^*$ is the function $F$ immediately after the game $\G_{e,x}$ is completed.
Our having won the game implies that when evaluating the circuit $\gamma_{e,x}$, while substituting the truth value of ``$z \in R_{H^*}$'' for any query of the form ``$z \in R_H$'', we have that $\gamma_{e,x} \neq L(x)$, which means that the reduction $\gamma_e$ does not output the correct value on input $x$ and thus $R_e$ is satisfied.
For any game $\G_{e^\prime,x^\prime}$ that is played later in the construction, by design $x^\prime$ is always chosen large enough so that any query that is not fixed during the pregame preprocessing has not appeared in any game that was played previously, so $\G_{e^\prime, x^\prime}$ will not conflict with $\G_{e,x}$ and $R_e$ will remain satisfied for the remainder of the construction.
\end{proof}

\begin{clm}
For any given requirement $R_e$, the $K$ player will only cheat on games $R_{e,x}$ for a finite number of strings $x$.
\end{clm}

\begin{proof}
If the $K$ player cheats on a game $R_{e,x}$, it means that he makes moves that causes his score to exceed $\epsilon = 2^{-e-3}$.
By the definition of how $K$'s moves are determined, this implies that $\sum_{z \in Z_{e,x}} 2^{-(K^\ts_{U_0}(z)-6)} \geq \epsilon$, so  $2^{-K^\ts_{U_0}(z)} \geq \epsilon/64$, where $Z_{e,x}$ is defined to be the set of all the queries that appear in the game $\G_{e,x}$ that are not fixed during the preprocessing stage.
However, for any two games $G_{e,x}$ and $G_{e,x^\prime}$ the sets $Z_{e,x}$ and $Z_{e, x^\prime}$ are disjoint, so if $K$ cheated on an infinite number of games associated with the requirement $R_e$, then this would imply that $\sum_{z \in \{0,1\}^*} 2^{K^\ts_{U_0}(z)} \geq \epsilon/64 + \epsilon/64 + \cdots$.  But this divergence would violate Proposition \ref{prefix_entr}.
\end{proof}

\begin{clm}\label{old4}
  During the construction, for any requirement $R_e$, we can always find a witness $x$ with the needed properties to construct $\G_{e,x}$.
\end{clm}

\begin{proof}
Suppose for some requirement $R_e$, our lexicographical search goes on forever without finding an $x$ such that $(1 - L(x^\prime)) \in val(\G_{e,x^\prime})$.
Then we will show that $L \in \PSPACE / \alpha(n)$, which is a contradiction.

  Here is the $\PSPACE$ algorithm to decide $L$ (using small advice).  
Hardcode all the answers for the initial sequence of strings up to the point where we got stuck in the construction.  Let $F^*$ be the function $F$ up to that point in the construction.  On a general input $x$, construct $\gamma_{e,x}$.  The advice function $\alpha(n)$ will give the truth-table of $\min(K^\ts_{U_0}(z)+5, F^*(z) + 3)$ for all queries $z$ such that $|z| \leq \log(\alpha(|x|))-1$.  For any query $z$ of $\gamma_{e,x}$ such that $|z| \leq \log(\alpha(|x|))-1$, fix the answer to the query according to the advice.

  If the $F$ player had a winning strategy for both the YES and NO player on game $\G_{e,x}$, then we wouldn't have gotten stuck on $R_e$.  Also the $F$ player must have a winning strategy for either the YES or the NO player, since he always has an advantage over the $K$ player when playing the game.  Therefore, because we got stuck, it must be that the $F$ player has a winning strategy for the YES player if and only if $L(x) = 1$.  Once the small queries have been fixed, finding the side (YES or NO) for which the $F$ player has a winning strategy on $\G_{e, x}$, and hence whether $L(x) = 1$ or $L(x) = 0$, can be done in $\PSPACE$.

To prove this, we will show that the predicate ``The $F$ player has a winning strategy as the YES player on $\G_{e, x}$'' can be computed in alternating polynomial time, which by \cite{cks} is equal to $\PSPACE$.
To compute this predicate, we must determine if \emph{for every} move of the $K$ player in round 1, there \emph{exists} a move for the $F$ player in round 1, such that \emph{for every} move of the $K$ player in round 2, there \emph{exists} a move for the $F$ player in round 2... such that when the game is finished the circuit $\gamma_{e,x}$ evaluates to TRUE. 
We can represent any state of the game (i.e. which of the polynomial number of queries have been fixed to be nonrandom so far, the score of the players, the current round, and whose turn it is) by a number of bits bounded by a polynomial in $|x|$.
Also, given a move by one of the players, it is easy to determine in polynomial time whether the move is legal and to compute the new score of the player after the move. (It suffices to add up a
polynomial number of rationals of the form $a/2^b$ where $b = n^{O(1)}$).
Also, because there are only a polynomial number of queries in the circuit $\gamma_{e,x}$, the total number of moves in the game is bounded by a polynomial.
Finally, evaluating the circuit at the end of the game can be done in polynomial time.
Thus the predicate in question can be computed in alternating polynomial time, which completes the proof.
\end{proof}

The following claim shows that Property 1 is satisfied.

\begin{clm}\label{time}
  $F(z)$ is computable in time $\ts^2(|z|)$.
\end{clm}

\begin{proof}
  The function $F$ is determined by the moves we play in games during the construction.  In order to prove the claim, we must show that if during the construction we as the $F$ player make a move that involves setting a string $z$ to be nonrandom, then fewer than $\ts^2(|z|)$ time steps have elapsed during the construction up to that point.  The machine $M$ that computes $F$ will on input $z$ run the construction for $\ts^2(|z|)$ steps.  If, at some point before this during the construction, we as the $F$ player make $z$ nonrandom, then $M$ outputs $|z|-4$.  Otherwise $M$ outputs $2|z|+3$.

  Suppose during the construction that we as the $F$ player make a move that sets a query $z$ to be nonrandom during a game $\G_{e,x}$.  Note that $|z| \geq \log(\alpha(|x|))$, otherwise $z$ would have been fixed during the preprocessing stage of the game.

  There are at most $2^{|x|+1}$ strings $x^\prime$ that we could have considered during our lexicographic search to find a game for which we had a winning strategy before finally finding $x$.  Let $s$ be the number of time steps that have elapsed during the construction before this search began.

  Let us first bound the amount of time it takes to reject each of these strings $x^\prime$.  To compute the circuit $\gamma_{e,x^\prime}$ takes at most $|x^\prime|^k$ time for some constant $k$.  For each query $y$ such that $|y| \leq \log(\alpha(|x^\prime|))-1$ we compute $\min(K^\ts_{U_0}(y)+5, F^*(y) + 3)$.  To calculate $F^*(y)$ it suffices to rerun the construction up to this point and check whether a move had been previously made on the string $y$.  
To do this takes $s$ time steps, and by construction we have that $\ts(|z|) \geq \ts(\log \alpha(|x|)) > |x| \geq |x^\prime| \geq \ts^4(s)$, so $s < |z|$.
By Proposition \ref{simpprop}, to compute $K^\ts_{U_0}(y)$ takes at most $2^{|y|}\ts(|y|) \leq 2^{|z|}\ts(|z|)$ time steps.  Therefore, since there can be at most $|x^\prime|^k$ such queries, altogether computing $\min(K^\ts_{U_0}(y)+5, F^*(y) + 3)$ for all these $y$ will take fewer than $|x^\prime|^k 2^{|z|}\ts(|z|)$ time steps.

  Then we must compute $L(x^\prime)$, and check whether $(1-L(x^\prime)) \in val(\G_{e,x^\prime})$.  Computing $L(x^\prime)$ takes $t_L(|x^\prime|)$ time.  By Claim \ref{old4}, once the small queries have been fixed appropriately, computing $val(\G_{e,x^\prime})$ can be done in $\PSPACE$, so it takes at most $2^{|x^\prime|^d}$ time for some constant $d$.

  Compiling all this information, and using the fact that for each of these $x^\prime$ we have that $|x^\prime| \leq |x|$, we get that the total number of timesteps needed to reject all of these $x^\prime$ is less than
$2^{|x|^{d^{\prime}}} 2^{|z|}t_L(|x|)\ts(|z|) $ for some constant $d^\prime$.

  During the actual game $\G_{e,x}$, before $z$ is made nonrandom the construction might have to compute $K_{U_0}^{\ts}(y) + 5$ for all queries of $\gamma_{e,x}$ for which $|y| \leq |z|$.  By Proposition \ref{simpprop} this takes at most $|x|^k 2^{|z|} \ts(|z|)$ time. 

  Therefore, overall, for some constant $d^{\prime\prime}$ the total amount of time steps elapsed before $z$ is made nonrandom in the construction is at most
  \[ T = 2^{|x|^{d^{\prime\prime}}} 2^{|z|} t_L(|x|)\ts(|z|) + s < \ts^2(|z|). \] 

  Here the inequality follows from the fact that $\ts(\log(\alpha(|x|))) > 2^{|x|^d}t_L(|x|)$ for any constant $d$, and that $|z| \geq \log(\alpha(|x|))$ . 
\end{proof}

Finally, to finish the proof of the theorem we need to show that Property 3 is satisfied.
\begin{clm}
$\sum_{x \in \{0,1\}^*} 2^{-H(x)} \leq \frac{1}{8}$.
\end{clm}

\begin{proof}To begin, notice that

\begin{align*}
\sum_{x \in \{0,1\}^*} 2^{-H(x)} &= \sum_{x \in \{0,1\}^*} 2^{-\min(K_{U_0}^\ts(x) + 5, F(x) + 3)} \leq \sum_{x \in \{0,1\}^*} 2^{-(K_{U_0}^\ts(x) + 5)} + \sum_{x \in \{0,1\}^*} 2^{-(F(x)+3)}.
\end{align*}
By Proposition \ref{prefix_entr},  $\sum_{x \in \{0,1\}^*} 2^{-K_{U_0}^\ts(x)} \leq 1$, so $\sum_{x \in \{0,1\}^*} 2^{-(K_{U_0}^\ts(x) + 5)} \leq 1/32$.
We also have that $\sum_{x \in \{0,1\}^*} 2^{-(F(x)+3)} = (1/8)\sum_{x \in \{0,1\}^*} 2^{-F(x)}$.
Therefore, it is enough that $\sum_{x \in \{0,1\}^*} 2^{-F(x)} \leq 1/2$, as this would imply that 
\[ \sum_{x \in \{0,1\}^*} 2^{-H(x)} \leq \frac{1}{32} + \frac{1}{8} \times \frac{1}{2} \leq \frac{1}{8}. \]

Let $Z_F$ be the set of all those queries that we (the $F$ player) make nonrandom during the construction by playing a move in one of the games.
We have that 
\begin{align*}
\sum_{x \in \{0,1\}^*} 2^{-F(x)} &= \sum_{x \in Z_F} 2^{-(|x|- 4)} + \sum_{x \not \in Z_F} 2^{-(2|x|+3)} \\
&= \sum_{x \in \{0,1\}^*} 2^{-(2|x|+3)} + \sum_{x \in Z_F} (2^{-(|x| - 4)} - 2^{-(2|x|+3)}) \\
&\leq  \frac{1}{8} + \sum_{x \in Z_F} (2^{-(|x| - 4)} - 2^{(2|x|+3)}).
\end{align*}

Thus it now suffices to show that $tot_F = \sum_{x \in Z_F} (2^{-(|x| - 4)} - 2^{(2|x|+3)}) \leq 1/4$.
Notice that $tot_F$ is exactly the total number of points that the $F$ player accrues in all games throughout the lifetime of the construction.
First let us consider those games on which the $K$ player cheats.
We know that in all these games, the $F$ player accrues fewer points than the $K$ player, and in particular accrues fewer points during these games than $tot_K$, the total number of points the $K$ player accrues in all games throughout the lifetime of the construction.
Let $Z_K$ be the set of all those queries that the $K$ player makes nonrandom during the construction by playing a move in one of the games.
We have that 
\begin{align*}
tot_K &= \sum_{z \in Z_K} 2^{-(|z| - 6)} - 2^{-(2|z|+3)} \leq \sum_{z \in Z_K} 2^{-(K_{U_0}^\ts(z) + 5)} \leq \sum_{z \in \{0,1\}^*} 2^{-(K_{U_0}^\ts(z) + 5)} \leq \frac{1}{32},
\end{align*}
where the first inequality uses that for all $z \in Z_K$, $K_{U_0}^\ts(z) \leq |z|-6$, and the last inequality again comes from Proposition \ref{prefix_entr}.

Now consider games on which $K$ does not cheat -- for each $R_e$ there will be exactly one of these.
On each of these games the $F$ player can accrue at most $\epsilon = 2^{-e-3}$ points.
Thus the total number of points the $F$ player accrues on all games that $K$ does not cheat on is at most $\sum_{e=1}^\infty 2^{-e-3} = 1/8$.

Therefore $tot_F \leq 1/32 + 1/8 \leq 1/4$.
\end{proof}

\section{Removing the Advice}\label{sec:no-advice}

With the plain Kolmogorov complexity function $C$, it is fairly clear what is meant by a ``time-efficient'' universal Turing machine.  Namely, $U$ is a time-efficient universal Turing machine if, for every Turing machine $M$, there is a constant $c$ so that, for every $x$, if there is a description $d$ for which $M(d) = x$ in $t$ steps, then there is a description $d'$ of length $\leq |d|+c$ for which $U(d')=x$ in at most $ct\log t$ steps.  However, with prefix-free Kolmogorov complexity, the situation is more complicated.  The easiest way to define universal Turing machines for the prefix-free Kolmogorov complexity function $K$ is in terms of {\em self-delimiting Turing machines.} These are machines that have one-way access to their input tape; $x$ is a valid input for such a machine if the machine halts while scanning the last symbol of $x$.  For such machines, the notion of time-efficiency carries over essentially unchanged.  However, there are several other ways of characterizing $K$ (such as in terms of partial-recursive functions whose domains form a prefix code, or in terms of prefix-free entropy functions).  The running times of the machines that give short descriptions of $x$ using some of these other conventions can be substantially less than the running times of the corresponding self-delimiting Turing machines.  This issue has been explored in detail by Juedes and Lutz \cite{juedes.lutz}, in connection with the $\Polytime$ versus $\NP$ problem.  Given that there is some uncertainty about how best to define the notion of time-efficient universal Turing machine for $\Ktime$-complexity, one possible response is simply to allow much more leeway in the time-efficiency requirement.

If we do this, we are able to get rid of the small amount of non-uniformity in our $\PSPACE$ upper bound.

\begin{defi}
  A prefix machine $U$ is an $f$-efficient universal prefix machine if there exist constants $c_M$ for each prefix machine $M$, such that
  \begin{enumerate}
  \item $\forall x$, $K_U(x) \leq K_M(x) + c_M$; and
  \item $\forall x$, $K^t_U(x) \leq K^{t'}_M(x) + c_M$ for all $t(n) > f(t'(n))$.
  \end{enumerate}
\end{defi}

In Definition \ref{timeeffdef} we defined a time-efficient universal prefix machine to be any $\Poly(n)$-efficient universal prefix machine.

\begin{defi}
  Define $\DTTRtimeb$ to be the class of languages $L$ such that for all computable $f$ there exists $t_0$ such that for all $f$-efficient universal prefix machines $U$ and $t \geq t_0$, $L \leq_{tt}^p R_{K^t_U}$.
\end{defi}

\begin{thm}\label{no_advice}
  $\BPP \subseteq \DTTRtimeb \subseteq \PSPACE \cap \Ppoly$.
\end{thm}

Note that $\DTTRtimeb \subseteq \DTTRtime$, so from Theorem \ref{ppolythm} we get $\DTTRtimeb \subseteq \Ppoly$.  Also, the proofs in \cite{bfkl} can be adapted to show that $\BPP \subseteq \DTTRtimeb$. So all we need to show is the $\PSPACE$ inclusion.

\begin{proof}[Proof of Theorem \ref{no_advice}]

  The proof is similar to the proof of Theorem \ref{main}, with some minor technical modifications.  Let $L$ be an arbitrary language from $\DTTRtimeb - \PSPACE$.  Because $\DTTRtimeb \subseteq \DTTRtime$, as in the proof of Theorem \ref{main} we have that $L$ is decidable in time $t_L < 2^{n^k} t^\prime(n^k)$ for some fixed time bound $t^\prime$ and constant $k$.

  Define $f$ to be a fast enough growing function that $f(n) > 2^{(t_L(n^d))^d}$ for any constant $d$, for all large $n$.  
  By the definition of $\DTTRtimeb$, for this $f$ there exists a $t_0$ such that for all $t \geq t_0$, $L \leq_{tt}^p R_{K^t_U}$. 
  Let $\ts(n)$ be a time bound such that for all $n$, $\ts(n) > f(n)$ and $\ts(n) > t_0(n)$.  To get our contradiction, we will show that there exists an $f$-efficient universal prefix machine $U$ and constant $c > 1$ such that $L \not \leq_{tt}^p R_{K^v_U}$, where $v(|x|) = 2^{ (t_L(\ts(|x|)))^c} > t_0(|x|)$.

  We will make use of the following revised proposition:

\begin{prop}[Revised Proposition \ref{old6}]\label{rev6}
  Let $U$ and $M$ be an $n^c$-efficient universal prefix Turing machine and a prefix Turing machine respectively.  Let $t,t^\prime$ be time bounds and $f,g$ be two time-constructible increasing functions, such that $g(n^{c}) < f(n)$.  Let $t^{\prime\prime}(|x|) = g(t(|x|)) = h(t^\prime(|x|))$.  Then there is an $f$-efficient universal prefix machine $U^{\prime}$ such that
  \[K_{U^{\prime}}^{t^{\prime\prime}}(x) = \min(K_U^t(x), K_{M}^{t^\prime}(x)) + 1.\]
\end{prop}

\begin{proof}
  Almost identical to before: On input $0y$, $U^{\prime}$ runs $U$ on input $y$.  If $U$ would output string $x$ on $y$ after $s$ steps, then $U^{\prime}$ outputs string $x$ after $g(s)$ steps.  Similarly, on input $1y$, $U^\prime$ runs $M$ on input $y$.  If $M$ would output string $x$ on $y$ after $s$ steps, then $U^{\prime}$ outputs string $x$ after $h(s)$ steps.

  Note that because $U$ is an $n^c$-efficient universal prefix machine, $U^\prime$ will be an $f$-efficient universal prefix machine. 
\end{proof}

We will construct functions $F$ and $H$ such that

\begin{enumerate}
\item $F$ is a total function such that for all $x$, $F(x) \leq 2|x|+3$, and $F(x)$ is computable in time $2^{(t_L(\ts(|x|)))^d}$ by a machine $M$ for some constant $d$.
\item $H(x) = \min(K_{U_0}^{\ts} + 5, F(x) + 3)$.
\item $\sum_{x \in \{0,1\}^*} 2^{-H(x)} \leq 1/8$
\item $L \not \leq ^p_{tt} R_H$
\end{enumerate}

\begin{clm}[Revised Claim \ref{old1}]\label{rev1}
  Given the above properties $H = K_{U}^v$ for some $f$-efficient universal prefix machine $U$ (which by Property 4 ensures that the theorem holds)
\end{clm}

\begin{proof}
  By Property 3 we have that $\sum_{x \in \{0,1\}^*} 2^{-F(x)+3} \leq 1/8$.  Therefore it holds that 
$$\sum_{x \in \{0,1\}^*} 2^{F(x)} \leq 1.$$  
Using this along with Property 1, we then have by Proposition \ref{coding} that $K^u_{M^\prime} = F+2$ for some prefix machine $M^\prime$ and constant $d^\prime$, where $u(x) = 2^{(t_L(\ts(|x|)))^{d^\prime}}$.  By Proposition \ref{old7} we have that $K^\ts_{U^\prime} = K_{U_0}^{\ts} + 4$ for some $n^{c^\prime}$-efficient universal prefix machine $U^\prime$.  Therefore, by Proposition \ref{rev6}, $H(x) = \min(K^{\ts}_{U_0}(x) + 5, F(x) + 3) = \min(K^{\ts}_{U^\prime}(x), K^u_{M'}(x)) + 1$ is $K^v_U$ for some $f$-efficient universal prefix machine $U$ and constant $c > 1$, where $v(|x|) = 2^{(t_L(\ts(|x|)))^c}$. (In this last step we are using the fact that $f(n) > 2^{(t_L(n^k))^k}$ for any constant $k$ to ensure that $U$ is an $f$-efficient universal prefix machine by Proposition $\ref{rev6}$).
\end{proof}

The construction is virtually the same as in Theorem \ref{main}.

There is one change from Theorem \ref{main} in how the games are played.  During the preprocessing step of a game $\G_{e,x}$, all queries $z$ such that $\ts(|z|) \leq |x|$ are fixed according to $\min(K^\ts_{U_0}(z)+5, F^*(z) + 3)$.

If we get stuck during our lexicographical search to find a suitable $x^\prime$ to play the game $\G_{e,x^\prime}$, then this implies that the language $L$ is in $\PSPACE$, since by Proposition \ref{simpprop}, for some constant $k$ fixing all queries $z$ such that $\ts(|z|) \leq |x|$ according to $\min(K^\ts_{U_0}(z)+5, F^*(z) + 3)$ can be done in $|x|^k 2^{|z|} \ts(|z|) \le |x|^k \ts(|z|)^2 \le |x|^{k+2}$ time (and then it is a $\PSPACE$ computation to determine which side the $F$ player has a winning strategy for).

It remains to prove the following claim.

\begin{clm}
  $F(z)$ is computable in time $2^{(t_L(\ts(|z|)))^d}$ for some constant $d$.
\end{clm}

\proof
  Suppose during the construction we as the $F$ player make a move that sets a query $z$ to be nonrandom during a game $\G_{e,x}$.  Note that $\ts(|z|) > |x|$, otherwise $z$ would have been fixed during the preprocessing stage of the game.

  As in the proof of Claim \ref{time}, we can bound the total amount of time steps elapsed before $z$ is made nonrandom in the construction to be at most
\[
T = 2^{|x|^d} 2^{|z|} t_L(|x|) \ts(|z|) + s < 2^{(t_L(\ts(|z|)))^d}\eqno{\qEd}
\]
And this concludes the proof of Theorem \ref{no_advice}.
\end{proof}

\section{Conclusion}

We have made some progress towards settling our research question in the case of time-bounded Kolmogorov complexity, but we have also discovered that this situation is substantially different from the ordinary Kolmogorov complexity. Solving this latter case will likely prove to be much harder.

We would like to prove an exact characterization, such as $\BPP = \DTTR$ (or the time-bounded analogue thereof), but there seems to be no naive way of doing this. It has been shown in \cite{bfkl} that the initial segment $R_{K}^{\le \log n}$, a string of length $n$, requires circuits of size $n/c$, for some $c > 1$ and all large $n$; it is this fact that is used to simulate $\BPP$.
However, much stronger circuit lower bounds for the initial segment do not seem to hold (cf. Theorems 4--9 of \cite{bfkl}), suggesting that $\RK$ has some structure.  This structure can actually be detected --- the reduction $M$ of Theorem \ref{thm:distinguish} can be adapted to distinguish $\RK$ from a random oracle w.h.p. --- but we still don't know of any way of using $\RK$ non-adaptively, other than as a pseudo-random string. A new idea will be needed in order to either prove or disprove the $\BPP = \DTTR$ conjecture.


\section*{Acknowledgments}

The first and third authors acknowledge NSF Grants CCF-0832787 and CCF-1064785.
The second author acknowledges NWO grant Networks.
The fourth author acknowledges FCT grant SFRH/BD/43169/2008.

\bibliographystyle{alpha}
\bibliography{refs}

\end{document}